\documentclass[12pt]{article}

\usepackage[T1]{fontenc}
\usepackage[utf8]{inputenc}
\usepackage[letterpaper,margin=1in]{geometry}
\usepackage{setspace}
\usepackage{amsmath,amssymb,amsthm,mathtools}
\usepackage[authoryear]{natbib}

\usepackage{tikz}
\usepackage[colorlinks=true,linkcolor=blue,citecolor=blue,urlcolor=blue]{hyperref}
\usepackage[disable]{todonotes}

\allowdisplaybreaks

\newtheorem{assumption}{Assumption}
\newtheorem{theorem}{Theorem}
\newtheorem{proposition}{Proposition}
\theoremstyle{remark}
\newtheorem{remark}{Remark}

\newcommand{\R}{\mathbb{R}}
\newcommand{\Prob}{\mathbb{P}}
\newcommand{\E}{\mathbb{E}}
\newcommand{\ind}{\mathbf{1}}

\DeclareMathOperator{\Supp}{Supp}

\title{Identification in Dynamic Dyadic Network \\Formation Models with Fixed Effects\footnote{We thank Ming Li, Yassine Sbai Sassi, and Zhengyan Xu for helpful discussions. AI tools (Claude, ChatGPT, Gemini, and Refine.ink) were used for research assistance and critical reviews; the authors assume full responsibility for any remaining errors.}}
\author{Wayne Yuan Gao\thanks{Department of Economics, University of Pennsylvania. Email: \texttt{waynegao@upenn.edu}}\,\, and Yi Niu\thanks{Department of Economics, University of Pennsylvania. Email: \texttt{n1y1@sas.upenn.edu}}}

\begin{document}
	
	\maketitle
	
	\begin{abstract}
		\noindent This paper establishes identification results in a dynamic dyadic network formation model with time-varying observed covariates, lagged local network statistics, and unobserved heterogeneity in the form of fixed effects. Our framework accommodates observed-covariate homophily, transitivity through common friends, second-order or indirect-friend effects, and more general local subgraph statistics within a single dynamic index model. The analysis combines two complementary ways of handling fixed effects: inequalities that integrate out time-invariant dyad heterogeneity by treating each dyad as a short panel, and signed-subgraph comparisons that difference out fixed effects algebraically through intertemporal variation within each dyad. We show that the semiparametric identifying restrictions can be sharpened using either or both of the following assumptions: (i) error distribution is serially independent with a known distribution, (ii) pairwise fixed effect takes the form of additive individual fixed effects. Combining (i) and (ii) under i.i.d.\ logit shocks, we obtain an exact conditional logit representation and provide sufficient conditions for point identification.
		
		\medskip
		
		\noindent \textbf{Keywords:} network formation, dynamic, dyadic, fixed effects, homophily, transitivity, subgraph, identification, semiparametric, conditional moment inequalities, logit
		
		\medskip
		
		\noindent \textbf{JEL Classification:} C14, C23, C31
		
	\end{abstract}
	
	\newpage
	
	\section{Introduction}
	
	Dynamic dyadic network formation models with state dependence, homophily, and local network spillovers create a natural tension between substantive realism and econometric tractability. On the one hand, lagged local network covariates such as common friends, friends-of-friends, and related subgraph counts are central for describing persistence, transitivity, and other forms of local clustering in network formation. On the other hand, once fixed effects are introduced, identification becomes difficult since observed link dynamics mix together structural state dependence, observed homophily, and time-invariant unobserved heterogeneity. An important econometric question in this context is whether these components can be separated using a panel of network data.
	
	The paper studies a dynamic dyadic network formation model in which current link surplus depends on time-varying observed dyadic covariates, a vector of lagged local network statistics, and time-invariant unobserved heterogeneity (fixed effects). The key insight is that, once the network statistics are lagged and observed, the model can be studied as a dynamic panel with lagged endogenous network covariates.
	
	Under unrestricted time-invariant dyad effects and an unknown error distribution, we propose two complementary semiparametric identification routes. The first route treats each dyad as a short panel and \emph{integrates out} the fixed effect with respect to an unknown distribution, while the second route uses dynamic signed-subgraph comparisons to \emph{difference} out fixed effects algebraically through intertemporal variation within each dyad. Both routes rely on a ``bounding-by-$c$'' technique proposed in \citet{GaoWang26} to handle the endogeneity issue arising from the lagged outcome variables. In addition, we synthesize the two routes into an umbrella framework that produces a general class of identifying restrictions along the ``difference-out'' / ``integrate-out'' spectrum. 
	
	The paper then shows how additional structure can be exploited to further sharpen the identification results. First, the assumption that errors are serially independent with a known distribution induces additional identifying restrictions on subgraphs with differenced-out fixed effects. Second, when pairwise fixed effect takes an additive form of individual-level fixed effects, we can obtain additional identifying restrictions using a weighted differencing argument. Third, we combine the two additional structures above with a logit error specification and show that the resulting model admits an exact conditional logit representation for any completely node-balanced configuration of edge-time cells---a class that includes within-date tetrads (the per-period analogue of \citealt{graham_2017}) but also intertemporal tetrads, triadic cycles, and other configurations that exploit both cross-node and cross-period variation. We provide sufficient conditions for point identification based on this enlarged class.
	
	\subsection*{Related Literature}
	
	Our paper builds upon and contributes to the econometric literature on network formation models. See, e.g., \citet{DEPAULA2020a,DEPAULA2020b} and \citet{GRAHAM2020}, for general surveys on this topic.
	
	More specifically, our paper belongs to the line of econometric work on \emph{dyadic} network formation models with homophily effects and individual unobserved heterogeneity (fixed effects), as pioneered by \citet{graham_2017}. \citet{graham_2017} provides the canonical  dyadic setup under the logit error specification. \citet{Candelaria2017}, \citet{toth2017}, \citet{Jochmans2018}, \citet{GAO2020}, and \citet*{GAO2023} consider various generalizations and adaptations of \citet{graham_2017}, but all focus on the static setting where the network is observed only once. 
	
	In contrast, this paper considers a dynamic environment in which current link formation depends on lagged local network statistics, following the conceptual framework of \citet{graham_2016}. Specifically, \citet{graham_2016} considers a dyadic network formation model with lagged common-friends transitivity, unrestricted dyad heterogeneity, and a fully i.i.d. logit shock specification. Its main identification device is the stable-neighborhood argument, designed to separate transitivity from unrestricted time-invariant dyad heterogeneity. However, \citet{graham_2016} does not incorporate observed covariates; in fact, once one introduces explicit time-varying observed homophily, the stable-neighborhood approach becomes less convenient, analogously to a similar issue in nonlinear dynamic panel models in \citet{honore2000panel}: one would need to compare dyads whose local network environments are sufficiently stable while simultaneously matching on time-varying covariate histories. Relative to \citet{graham_2016}, our contribution is to develop identification tools that remain applicable once explicit time-varying observed homophily is brought into the model. Furthermore, we provide results not only for the parametric setting with i.i.d. logit setup, but also for a semiparametric setting where errors are allowed to be serially correlated with unknown distributions.
	
	The paper also draws directly on \citet{GaoWang26}, which develops panel-style ``bounding-by-$c$'' arguments for nonlinear dynamic models with fixed effects. The model setup in this paper is analogous to a nonlinear panel model with lagged endogenous regressors. However, in the current paper, at each time point, the "cross-sectional" data structure is given by a network of individuals along with their covariates, which is different from the ``purely individual'' data structure in \citet{GaoWang26}. Hence, while the core idea of \citet{GaoWang26} continues to be useful, our paper considers a data structure not covered in \citet{GaoWang26}, exploits nontrivial adaptations of the ``bounding-by-$c$'' technique, and obtains identifying restrictions that have no direct analog in the standard panel data setting.
	
	The paper is also related to and different from \citet*{GaoLiXu26}, which studies static strategic network formation models. First, \citet*{GaoLiXu26} considers a data structure where a single large network is observed once, while our current paper focuses on the alternative ``panel'' data structure where we have network data over multiple time periods. The time dimension in our current paper allows us to carry out intertemporal comparisons that have no direct analog in  \citet*{GaoLiXu26}. Second, both \citet*{GaoLiXu26} and this paper provide econometric methods to study how local network structure affects the linking decision between two individuals, but the two papers approach this issue from two very different, and likely complementary, perspectives: \citet*{GaoLiXu26} considers strategic interactions and simultaneity issues in a static setting, while the current paper considers a sequentially exogenous setup based on lagged networks. One implication is that, in our current paper, there is no need to impose separate subnetwork-CCP identifiability conditions as required in \citet*[Assumption 4 and Section 4]{GaoLiXu26}. Third, while both papers exploit signed-subgraph and weighted-differencing techniques to eliminate fixed effects, the current paper features results with no analogues in \citet*{GaoLiXu26}, since here we can exploit intertemporal variations and the ``bounding-by-$c$'' technique from \citet{GaoWang26}, and obtain results even without the additive fixed effect structure, which is always assumed in \citet*{GaoLiXu26}.

	The rest of the paper proceeds as follows. Section \ref{Section: Model Setup} introduces the model setup. Section \ref{Section: Semiparametric Identification} develops the paper's main semiparametric identification architecture under arbitrary dyad effects, including both dyad-panel and dynamic signed-subgraph arguments and the unified partial-differencing perspective linking them. Section~\ref{Section: Sharper Identification under Additional Structures} studies how additional structure sharpens those results through known composite-error distributions and additive-node restrictions. Section~\ref{Section: Conclusion} concludes.
	
	\section{Model Setup}\label{Section: Model Setup}
	
	This section introduces the paper's baseline dynamic dyadic network formation model and the notation used throughout the identification analysis.
	
	Consider a set of \emph{nodes} (representing individuals or other types of economic agents) indexed by $i$ with dyads, i.e., pairs of nodes, indexed by $ij$. Throughout this paper, we focus on undirected and unweighted networks. Writing $D_{ijt}\in\{0,1\}$ as the link indicator for dyad $ij$ at time $t$, we consider the following dynamic network formation model
	\begin{equation}
		D_{ijt}
		=
		\ind\left\{
		Z_{ijt}'\alpha_0 + X_{ij,t-1}'\lambda_0 + A_{ij} - U_{ijt} \ge 0
		\right\},
		\quad t=1,\ldots,T,
		\label{eq:main-model}
	\end{equation}
	with $Z_{ijt} := |Z_{it}-Z_{jt}| \in \R^{d_h}$  denoting the observed time-varying dyadic covariates at time $t$, where the node-level covariate $Z_{it}$ may be vector-valued and $|\cdot|$ denotes coordinate-wise absolute value. Here $X_{ij,t-1}\in\R^{d_x}$ is a vector of observed lagged network covariates of fixed dimension, $A_{ij}$ is a time-invariant unobserved dyad fixed effect, and $U_{ijt}$ are idiosyncratic time-varying dyadic shocks. The unknown parameter vector $\theta_0 := (\alpha_0',\lambda_0')'$ consists of the coefficient vector on observed homophily $\alpha_0\in\R^{d_h}$ and that on lagged network covariates $\lambda_0\in\R^{d_x}$.\footnote{Because the error distribution is left unspecified in the semiparametric analysis, the model is invariant to a common positive rescaling of $(\theta,A_{ij},U_{ijt})$. The semiparametric identified sets derived below fully reflect this scale indeterminacy. Scale is pinned once the error distribution is specified, as in the logit specification of Section~\ref{Section: Sharper Identification under Additional Structures}.}
	
	Note that any time-invariant dyadic observable is absorbed by the fixed effect $A_{ij}$ in the unrestricted-dyad-effects baseline; the semiparametric identification arguments therefore exploit variation in the time-varying covariates $Z_{ijt}$ and the lagged network statistics $X_{ij,t-1}$.
	
	The framework incorporates several familiar ingredients in network formation models. Since $Z_{ijt}$ is constructed as distances between node-level observed characteristics, the model captures homophily with respect to observed characteristics. If $X_{ij,t-1}$ includes lagged common friends, the model captures potential preference for transitivity. If $X_{ij,t-1}$ includes lagged friends-of-friends or other second-order reachability measures, it captures indirect-friend effects. More generally, $X_{ij,t-1}$ may collect any fixed-dimensional vector of lagged local subgraph statistics that a researcher deems relevant for the network formation problem.
	
	It is also useful to explicitly relate our model to the setup in \citet{graham_2016}, whose baseline dynamic specification is
	\begin{equation}
		D_{ijt}
		=
		\ind\left\{
		\beta_0 D_{ij,t-1} + \gamma_0 R_{ij,t-1} + A_{ij} - U_{ijt} \ge 0
		\right\},
		\label{eq:graham-original}
	\end{equation}
	where $R_{ij,t-1} := \sum_{k \neq i,j} D_{ik,t-1}D_{jk,t-1}$ is the lagged number of common friends. Note that equation \eqref{eq:graham-original} is a special case of \eqref{eq:main-model}, obtained by omitting the $Z_{ijt}'\alpha_0$ term and setting
	\[
	X_{ij,t-1}
	:=
	\bigl(D_{ij,t-1},R_{ij,t-1}\bigr)',
	\quad
	\lambda_0 := (\beta_0,\gamma_0)'.
	\]
	Our framework is therefore broader in two directions at once: it allows explicit observed-covariate homophily through time-varying $Z_{ijt}$ and it allows a general fixed-dimensional vector of lagged local network covariates rather than only lagged own-link status and common friends. The current model also contains the static formation model of \citet{graham_2017} as an effectively nested special case, which can be obtained by suppressing the lagged-network vector $X_{ij,t-1}$, restricting the fixed effect to take the additive-node form $A_{ij}=\nu_i+\nu_j$, and interpreting the resulting model at a single time point. Nothing in the semiparametric arguments below uses the special two-regressor form $(D_{ij,t-1},R_{ij,t-1})$ beyond the fact that it is an observed lagged vector that satisfies certain exogeneity conditions, and the proofs go through unchanged for any fixed-dimensional $X_{ij,t-1}$.
	
	\paragraph{Observed data.}
	The econometrician observes the node-level covariates $(Z_{it})_{t=1}^{T}$ for each node $i$ and the network $(D_{ijt})_{t=0}^{T}$ for all dyads $ij$. Because $X_{ij,t-1}$ is computed from the lagged network, its construction at $t=1$ requires the initial network $(D_{ij0})_{ij}$, which is treated as given. No distributional assumption is placed on the initial network.
	
	\bigskip
	
	In the following, it would be convenient to write $\theta := (\alpha',\lambda')'$ and
	\[
	W_{ijt}(\theta)
	:=
	Z_{ijt}'\alpha + X_{ij,t-1}'\lambda,
	\quad
	V_{ijt} := U_{ijt} - A_{ij},
	\]
	so that model \eqref{eq:main-model} becomes
	\[
	D_{ijt} = \ind\{V_{ijt} \le W_{ijt}(\theta_0)\}.
	\]
	From the viewpoint of dyad $ij$, the model is therefore a dynamic binary panel with one time-invariant dyad effect and lagged endogenous network covariates. Below we explain how to exploit the intertemporal variations of the panel structure, as well as the additional two-dimensional network structure at each fixed time point, to obtain identifying restrictions.
	
	\section{Semiparametric Identification}\label{Section: Semiparametric Identification}
	
	This section develops the paper's semiparametric identification approach under unrestricted form of dyad fixed effects. The first subsection integrates the fixed effect out by treating each pair as a short panel. The second subsection differences the dyad effect out directly through dynamic signed-subgraph comparisons. The third shows that these are two endpoints of a broader spectrum that combines differencing and integration.
	
	\begin{assumption}[Idiosyncratic Dyadic Shocks]
		\label{ass:eh}
		Write $U_{ij}^{1:T}: = (U_{ij1},\dots,U_{ijT})'$ and $Z_i^{1:T}: = (Z_{i1}',\dots,Z_{iT}')'$. The dyad-level shock vectors $\{U_{ij}^{1:T}: i<j\}$ are i.i.d. across dyads and are jointly independent of $(A_{ij})_{ij}$ and $(Z_i)_i$, i.e.,
		\[
		\{U_{ij}^{1:T}: i<j\}
		\perp
		\left(
		\{A_{ij}: i<j \in\{1,...,n\}\},
		\{Z_i^{1:T}: i\in\{1,\ldots,n\}\}
		\right).
		\]
		Moreover, the distribution of $U_{ijt}$ is homogeneous across time $t$, i.e., for each dyad $ij$ and each pair of dates $t,s \in \{1,\ldots,T\}$, $U_{ijt} \sim U_{ijs}.$
	\end{assumption}
	
	Throughout, all conditional distributions are assumed to admit regular versions, so that conditioning on exact realizations of covariate histories and taking suprema or infima over their supports are well-defined operations.\footnote{Equivalently, the reader may interpret all sup/inf operations as essential suprema/infima with respect to the relevant marginal measures.}
	
	Assumption \ref{ass:eh} is standard in the dyadic network formation literature. It says that the dyad-level shock process is i.i.d.\ across dyads, exogenous relative to both the time-invariant latent heterogeneity and the entire observed exogenous covariate array, has homogeneous marginals over time, and may nevertheless be serially correlated within a dyad. Arbitrary dependence between $A_{ij}$ and the covariate histories is still allowed. The i.i.d.\ assumption across dyads rules out unobserved community-level shocks that simultaneously affect multiple dyads at the same date; such extensions are left to future work. The i.i.d.\ logit assumption in \citet{graham_2016} can be viewed as a strengthening of Assumption \ref{ass:eh}.
	
	\subsection{Dyadic Panel Identification}
	
	We apply the ``bounding-by-$c$'' technique in \citet{GaoWang26} and obtain bounds free of lagged outcome variables, which allows us to exploit the independence and time-homogeneity assumption on idiosyncratic dyadic shocks $U_{ijt}$. Specifically, fix $h\in\Supp(Z_{ij}^{1:T})$, $c\in\R$, and pair of dates $(t,s)$. If $D_{ijt}=1$ and $W_{ijt}(\theta_0)\le c$, then by \eqref{eq:main-model},
	\[
	V_{ijt}\le W_{ijt}(\theta_0)\le c \implies D_{ijt}\ind\{W_{ijt}(\theta_0)\le c\}
	\le
	\ind\{V_{ijt}\le c\}.
	\]
	Taking expectations conditional on $Z_{ij}^{1:T}=h$ gives
	\[
	L_t(c \mid h; \theta)
	:=\E\left[D_{ijt}\ind\{
	W_{ijt}(\theta)\le c\}
	\mid
	Z_{ij}^{1:T}=h
	\right] \le \Prob(V_{ijt}\le c \mid Z_{ij}^{1:T}=h).
	\]
	Similarly, if $D_{ijs}=0$ and $W_{ijs}(\theta_0)\ge c$, one can get
	\[
	U_s(c \mid h; \theta)
	:=1-\E\!\left[
	(1 - D_{ijs})\ind\{W_{ijs}(\theta)\ge c\}
	\mid
	Z_{ij}^{1:T}=h
	\right]
	\ge
	\Prob(V_{ijs}\le c \mid Z_{ij}^{1:T}=h).
	\]
	By the joint independence and the homogeneous-marginal parts of Assumption \ref{ass:eh}, $\Prob(V_{ijt}\le c \mid Z_{ij}^{1:T}=h)$ is common across dates. After taking supremum over $t$ and infimum over $s$, we obtain an identified set for $\theta$. We summarize the results in the following proposition.
	
	\begin{proposition}[Dyadic Panel Identifying Restrictions]
		\label{prop:dyad}
		For any $\theta=(\alpha',\lambda')'$, define the intertemporally aggregated bounds
		\[
		\overline L(c \mid h; \theta)
		:=
		\max_{t=1,\ldots,T} L_t(c \mid h; \theta),
		\quad
		\underline U(c \mid h; \theta)
		:=
		\min_{t=1,\ldots,T} U_t(c \mid h; \theta).
		\]
		Then under \eqref{eq:main-model} and Assumption \ref{ass:eh}, we have 
		$\theta_0 \in \Theta_I^{\mathrm{dyad}}$, where
		\[
		\Theta_I^{\mathrm{dyad}}
		:=
		\left\{
		\theta:
		\overline L(c \mid h; \theta)
		\le
		\underline U(c \mid h; \theta)
		\text{ for all } c\in\R
		\text{ and all } h\in\Supp(Z_{ij}^{1:T})
		\right\}.
		\]
	\end{proposition}
	
	\begin{remark}[About Sharpness]
		\label{rem:sharpness}
		Proposition \ref{prop:dyad} shows that $\theta_0$ belongs to the displayed restriction set, but it does not claim that the set is sharp. Throughout this paper, we use ``identified set'' in this standard sense without claiming sharpness. Establishing sharpness in the present dynamic-network environment appears substantially harder and is left to future work.
	\end{remark}
	
	\begin{remark}[Role of the time dimension]
		\label{rem:time-dimension}
		All semiparametric results in this section require at least $T\ge 2$ time periods, since the identifying restrictions compare outcomes across distinct dates. A larger $T$ enlarges the class of available comparisons: additional dates contribute to the maximum over $t$ and minimum over $s$ in Proposition~\ref{prop:dyad}, and enlarge the class of admissible balanced signed subgraphs in Propositions~\ref{prop:dyad-transition}--\ref{prop:signed-subgraph}. This does not automatically imply monotone shrinkage of the identified set, since the conditioning objects also grow with $T$, but it does expand the set of identifying restrictions that can be brought to bear.
	\end{remark}

	\subsection{Signed Subgraph Identification}
	
	The signed-subgraph approach is closer to \cite*{GaoLiXu26}. It uses time as an additional differencing dimension and constructs events over edge-time cells so that fixed effects cancel algebraically. Because the network regressors are lagged, one can compare edge-time cells without confronting contemporaneous simultaneity. The key point is that the propositions below use only the exogeneity part of Assumption \ref{ass:eh}; they do \emph{not} use homogeneous marginals and therefore remain valid under arbitrary serial correlation. We begin with the smallest nontrivial case, a two-period transition for one dyad, and then state the general signed-subgraph version.

	\begin{proposition}[Dyad-transition inequalities]
		\label{prop:dyad-transition}
		Fix two dates $t \neq s$ and define
		\[
		\Delta_{ts}W_{ij}(\theta)
		:=
		W_{ijt}(\theta)-W_{ijs}(\theta),
		\quad
		\Delta_{ts}U_{ij}
		:=
		U_{ijt}-U_{ijs}, \quad \mathcal Z_{ij}^{1:T}
		:=
		\bigl({Z_i^{1:T}}',{Z_j^{1:T}}'\bigr)'.
		\]
		Under \eqref{eq:main-model} and Assumption \ref{ass:eh}, for every $c\in\R$ and every $z$ in the support of $\mathcal Z_{ij}^{1:T}$,
		\[
		\E\left[
		D_{ijt}(1-D_{ijs})
		\ind\{\Delta_{ts}W_{ij}(\theta_0)\le c\}
		\mid
		\mathcal Z_{ij}^{1:T}=z
		\right]
		\le
		\Prob\bigl(\Delta_{ts}U_{ij}<c\bigr),
		\]
		and
		\[
		\E\left[
		(1-D_{ijt})D_{ijs}
		\ind\{\Delta_{ts}W_{ij}(\theta_0)\ge c\}
		\mid
		\mathcal Z_{ij}^{1:T}=z
		\right]
		\le
		\Prob\bigl(\Delta_{ts}U_{ij}>c\bigr).
		\]
		Consequently,
		\[
		\begin{aligned}
			\sup_z\,
			\E\left[
			D_{ijt}(1-D_{ijs})
			\ind\{\Delta_{ts}W_{ij}(\theta_0)\le c\}
			\mid
			\mathcal Z_{ij}^{1:T}=z
			\right]
			\le\;&
			\inf_z
			\Bigl[
			1-
			\E\left[
			(1-D_{ijt})D_{ijs}
			\right.\\
			&\left.
			\ind\{\Delta_{ts}W_{ij}(\theta_0)\ge c\}
			\mid
			\mathcal Z_{ij}^{1:T}=z
			\right]
			\Bigr].
		\end{aligned}
		\]
	\end{proposition}

	Proposition \ref{prop:dyad-transition} is the simplest dynamic analog of the Gao-Li-Xu differencing logic. The dyad effect $A_{ij}$ appears once with a positive sign and once with a negative sign, so it cancels exactly. The conditioning is only on exogenous $Z$ histories; the lagged network vector $X_{ij,t-1}$ remains inside the random index difference and need not be conditioned on. Call a triple $(i,j,t)$ with $i<j$ and $t\in\{1,\ldots,T\}$ an \emph{edge-time cell}. For any finite collection $\mathcal C$ of edge-time cells, let $N(\mathcal C)$ denote the set of nodes appearing in $\mathcal C$, and define the corresponding exogenous history vector by
	\[
	\mathcal Z_{\mathcal C}^{1:T}
	:=
	\bigl({Z_m^{1:T}}'\bigr)_{m\in N(\mathcal C)}'.
	\]

	\begin{proposition}[Dynamic signed-subgraph inequalities]
		\label{prop:signed-subgraph}
		Let $\mathcal C^{+}$ and $\mathcal C^{-}$ be nonempty finite collections of edge-time cells. Suppose they are \emph{balanced} in the sense that for every dyad $(i,j)$,
		\[
		\#\{t:(i,j,t)\in\mathcal C^{+}\}
		=
		\#\{t:(i,j,t)\in\mathcal C^{-}\}.
		\]
		Define
		$$Y_{\mathcal C}^+ : = \prod_{e\in\mathcal C^+} D_e
		\prod_{e\in\mathcal C^-} (1-D_e),\quad Y_{\mathcal C}^- : =\prod_{e\in\mathcal C^+} (1-D_e)
		\prod_{e\in\mathcal C^-} D_e,$$
		where $D_e$ denotes the link indicator attached to cell $e$. Also define
		\[
		\Delta_{\mathcal C}W(\theta)
		:=
		\sum_{(i,j,t)\in\mathcal C^{+}} W_{ijt}(\theta)
		-
		\sum_{(i,j,t)\in\mathcal C^{-}} W_{ijt}(\theta),\quad \Delta_{\mathcal C}U
		:=
		\sum_{(i,j,t)\in\mathcal C^{+}} U_{ijt}
		-
		\sum_{(i,j,t)\in\mathcal C^{-}} U_{ijt}.
		\]
		Under \eqref{eq:main-model} and Assumption \ref{ass:eh}, for every $c\in\R$,
		\[
		\sup_z
		\E\left[
		Y_{\mathcal C}^+
		\ind\{\Delta_{\mathcal C}W(\theta_0)\le c\}
		\mid
		\mathcal Z_{\mathcal C}^{1:T}=z
		\right]
		\le
		\inf_z
		\left[
		1-
		\E\left[
		Y_{\mathcal C}^-
		\ind\{\Delta_{\mathcal C}W(\theta_0)\ge c\}
		\mid
		\mathcal Z_{\mathcal C}^{1:T}=z
		\right]
		\right].
		\]
	\end{proposition}

	\begin{remark}
		\label{rem:signed-subgraph}
		Proposition \ref{prop:signed-subgraph} is the direct dynamic analog of the Gao-Li-Xu subgraph argument. Proposition \ref{prop:dyad-transition} is its two-cell special case, obtained by taking $\mathcal C^+=\{(i,j,t)\}$ and $\mathcal C^-=\{(i,j,s)\}$. Dyad transitions are the smallest balanced signed subgraphs, but richer dynamic objects are possible. One can mix cross-sectional and intertemporal differencing in the same construction, provided each dyad appears with net sign zero. Because the network covariates are lagged, the isolation arguments that are required in simultaneous static strategic models are not necessary here. Unlike Proposition \ref{prop:dyad}, these signed-subgraph inequalities do not use homogeneous marginals over time. Like Proposition \ref{prop:dyad}, they continue to hold under arbitrary serial correlation of the pairwise shock process.
		
		Analogously to Proposition \ref{prop:dyad}, define
		\[
		\Theta_I^{\mathrm{subgraph}}
		:=
		\left\{
		\theta:
		\text{the inequality in Proposition \ref{prop:signed-subgraph} holds for all balanced } (\mathcal C^+,\mathcal C^-) \text{ and all } c\in\R
		\right\}.
		\]
		Then $\theta_0\in\Theta_I^{\mathrm{dyad}}\cap\Theta_I^{\mathrm{subgraph}}$, and the two identified sets are in general not nested.
	\end{remark}
	
	\subsection{Unified Partial-Differencing Perspective}
	
	The two semiparametric approaches above can be viewed as extreme points of a broader spectrum:
	\[
	\begin{tikzpicture}[x=1cm,y=1cm,baseline=(current bounding box.center)]
		\draw[thick] (-5,0) -- (5,0);
		\fill (-5,0) circle (1.7pt);
		\fill (5,0) circle (1.7pt);
		\node[below=7pt,align=center,text width=2.8cm] at (-5,0) {complete\\differencing out};
		\node[below=7pt,align=center,text width=2.8cm] at (5,0) {complete\\integrating\\ out};
		\node[above=8pt,align=center,text width=7.2cm] at (0,0)
		{partial differencing / partial integration};
	\end{tikzpicture}
	\]
	The basic object is a signed comparison over edge-time cells in which some fixed-effect components cancel algebraically, while the remaining components are absorbed into a common latent CDF.
	
	The clean economic interpretation is exactly a split between two roles: First, the \emph{differenced-out parts}. These are the dyad components whose fixed effects cancel algebraically. For them, one only needs the exogeneity part of Assumption \ref{ass:eh}. Their exogenous histories may therefore be conditioned on freely and then profiled out through sup/inf operations. Second, the \emph{integrated-out parts}. These are the dyad components whose fixed effects do not cancel. For them, one relies on the homogeneity/common-law part of Assumption \ref{ass:eh}. Their residual contribution is absorbed into a latent CDF that is held fixed while one takes envelopes over admissible comparison objects.
	
	Define a comparison object as an ordered pair $g: = (\mathcal C_g^+,\mathcal C_g^-)$, where $\mathcal C_g^+$ and $\mathcal C_g^-$ are finite collections of edge-time cells indexed by $g$. Define
	\[
	Y_g^+
	:=
	\prod_{e\in\mathcal C_g^+} D_e
	\prod_{e\in\mathcal C_g^-} (1-D_e),
	\quad
	Y_g^-
	:=
	\prod_{e\in\mathcal C_g^+} (1-D_e)
	\prod_{e\in\mathcal C_g^-} D_e,
	\]
	\[
	\Delta_gW(\theta)
	:=
	\sum_{e\in\mathcal C_g^+} W_e(\theta)
	-
	\sum_{e\in\mathcal C_g^-} W_e(\theta),
	\quad
	\Delta_gU
	:=
	\sum_{e\in\mathcal C_g^+} U_e
	-
	\sum_{e\in\mathcal C_g^-} U_e.
	\]
	For each dyad $(i,j)$, let
	\[
	\rho_g(i,j)
	:=
	\#\{t:(i,j,t)\in\mathcal C_g^+\}
	-
	\#\{t:(i,j,t)\in\mathcal C_g^-\}.
	\]
	Also, define the residual-dyad set and the vector of dyadic-covariate histories for the uncanceled dyads
	\[
	\mathcal R_g := \{(i,j):\rho_g(i,j)\neq 0\},\quad Z_{\mathcal R_g}^{1:T}
	:=
	\bigl({Z_{ij}^{1:T}}'\bigr)'_{(i,j)\in\mathcal R_g}.
	\]
	
	Assume that for each $g\in\mathcal G$, both $\mathcal C_g^+$ and $\mathcal C_g^-$ are nonempty. On the event $Y_g^+=1$,
	\[
	\Delta_gU
	<
	\Delta_gW(\theta_0)
	+
	\sum_{(i,j)\in\mathcal R_g}\rho_g(i,j)A_{ij},
	\]
	and on $Y_g^-=1$ the reverse strict inequality holds. Thus, if one defines
	\[
	M_g
	:=
	\Delta_gU
	-
	\sum_{(i,j)\in\mathcal R_g}\rho_g(i,j)A_{ij},
	\]
	then $Y_g^+=1$ implies $M_g<\Delta_gW(\theta_0)$ and $Y_g^-=1$ implies $M_g>\Delta_gW(\theta_0)$.
	
	\begin{proposition}[Partial-differencing envelope within a fixed residual-load class]
		\label{prop:partial-spectrum}
		Let $\mathcal G$ be a family of comparison objects $g$ such that:
		\begin{enumerate}
			\item all $g\in\mathcal G$ have the same residual-load vector $\rho_g=\rho$ and hence the same residual-dyad set $\mathcal R$;
			\item for each $g\in\mathcal G$, one can partition the observable exogenous histories entering $g$ into a retained component $S_g$ and a nuisance component $T_g$;
			\item the retained component is common across $g\in\mathcal G$, in the sense that $S_g=S_{\mathcal R}$ for all $g$, where $S_{\mathcal R}$ is built from the exogenous histories of the uncanceled dyads in $\mathcal R$;
			\item conditional on $S_{\mathcal R}=s$, the distribution of $M_g$ is common across $g\in\mathcal G$ and does not depend on $T_g$.
		\end{enumerate}
		Then for every $c\in\R$ and every $s$ in the support of $S_{\mathcal R}$, there exists a CDF $F(c\mid s)$ such that
		\[
		\sup_{g\in\mathcal G}
		\sup_{t\in\Supp(T_g\mid S_{\mathcal R}=s)}
		\E\left[
		Y_g^+
		\ind\{\Delta_gW(\theta_0)\le c\}
		\mid
		S_{\mathcal R}=s,\;
		T_g=t
		\right]
		\le
		F(c\mid s),
		\]
		and
		\[
		F(c\mid s)
		\le
		\inf_{g\in\mathcal G}
		\inf_{t\in\Supp(T_g\mid S_{\mathcal R}=s)}
		\left[
		1-
		\E\left[
		Y_g^-
		\ind\{\Delta_gW(\theta_0)\ge c\}
		\mid
		S_{\mathcal R}=s,\;
		T_g=t
		\right]
		\right].
		\]
	\end{proposition}

	\begin{remark}
		\label{rem:partial-spectrum}
		Proposition \ref{prop:partial-spectrum} provides a taxonomic framework that nests the two semiparametric approaches developed above, but it does so \emph{within a fixed residual-load class}. That is, the proposition pools only over comparison objects that leave the same uncanceled dyads with the same residual coefficients. Different residual-load vectors generate different retained conditioning objects and therefore different envelope inequalities; the overall identified set is obtained by intersecting the restrictions from those separate classes. 
		
		First, the \emph{complete integration out} class. Take $\mathcal G=\{g_t:t=1,\ldots,T\}$, where $g_t$ is the one-cell comparison object built from edge-time cell $(i,j,t)$, so that $\mathcal C_{g_t}^+=\{(i,j,t)\}$ and $\mathcal C_{g_t}^-=\varnothing$. Then $\rho(i,j)=1$, so no dyad effect is canceled. Set $S_{\mathcal R}=Z_{ij}^{1:T}$ and let $T_g$ be empty. The common conditional law of $M_{g_t}=U_{ijt}-A_{ij}$ given $Z_{ij}^{1:T}=h$ follows from the joint independence and the homogeneous-marginal part of Assumption \ref{ass:eh}. Although these one-cell comparison objects have $\mathcal C_{g_t}^-=\varnothing$ and hence fall outside the strict-inequality setting of the proposition, the resulting weak-inequality bounds are still valid and recover Proposition \ref{prop:dyad}.
		
		Second, the \emph{complete differencing out} class. Take $\mathcal G=\{g\}$ with $\mathcal C_g^+=\{(i,j,t)\}$ and $\mathcal C_g^-=\{(i,j,s)\}$. Then $\rho\equiv 0$, so the dyad effect is canceled completely and $M_g=\Delta_{ts}U_{ij}$. Set $S_{\mathcal R}$ to be degenerate and let $T_g=\mathcal Z_{ij}^{1:T}$. This gives Proposition \ref{prop:dyad-transition}. More generally, any balanced signed subgraph has $\rho\equiv 0$ and falls under Proposition \ref{prop:signed-subgraph}.
		
		Lastly, the \emph{partial differencing / partial integration} class. For any fixed class of intermediate comparison objects with the same residual-load vector $\rho$, the zero-load dyads are differenced out, while the nonzero-load dyads are integrated out through the unknown CDF $F(c\mid s)$. This provides an organizing perspective under arbitrary dyad effects.
		
		In the notation of Proposition \ref{prop:partial-spectrum}, $T_g$ should be read as the exogenous histories attached to the differenced-out pieces, while $S_{\mathcal R}$ should be read as the exogenous histories attached to the absorbed pieces. Exogeneity lets one condition on $T_g$ and then profile over it, whereas homogeneity/common-law restrictions are used to compare the latent CDF indexed by $S_{\mathcal R}$ across comparison objects. This unified view explains why the dyad-panel and signed-subgraph approaches are complementary rather than redundant. The dyad-panel approach sits at the ``fully integrated'' end of the spectrum, while the signed-subgraph approach sits at the ``fully differenced'' end. Intermediate partial-differencing designs lie between those extremes, but each residual-load class contributes its own envelope inequality rather than all classes pooling into a single common CDF. The later strengthenings below move along the same spectrum by making some composite-error CDFs explicit or by enlarging the class of admissible partial-differencing designs.
	\end{remark}
	
	\section{Sharper Identification under Additional Structures}\label{Section: Sharper Identification under Additional Structures}
	
	Section \ref{Section: Semiparametric Identification} imposed neither parametric knowledge of the shock process nor additional structure on $A_{ij}$. Two strengthenings are especially useful. First, if the common marginal CDF of $U_{ijt}$ is known and the shock process is serially independent, then every fully differenced comparison has a known composite-error CDF and the bounding inequalities become explicit. Second, the additive-node structure $A_{ij}=\nu_i+\nu_j$ enlarges the class of valid weighted-differencing arguments even when the CDF is unknown.
	
	\subsection{Known Marginal CDF and Serial Independence}
	
	Suppose now that the common marginal CDF of $U_{ijt}$ is known, continuous, and denoted by $F_U$. Suppose in addition that the shock process $U_{ijt}$ is serially independent within each dyad $(i,j)$. Because Assumption \ref{ass:eh} already gives i.i.d.\ shock vectors across dyads, this strengthening implies independence across all distinct edge-time cells. Therefore every differencing design that fully removes the relevant fixed effects produces a composite error with known CDF.
	
	\begin{proposition}[Explicit bounds under a known marginal CDF and serial independence]
		\label{prop:known-composite}
		Suppose Assumption \ref{ass:eh} holds, the common marginal CDF $F_U$ of $U_{ijt}$ is known and continuous, and $U_{ij1},\ldots,U_{ijT}$ are independent for every dyad $(i,j)$.
		\begin{enumerate}
			\item[(1)] For any pair of dates $t\neq s$, define
			\[
			\Delta_{ts}U_{ij}
			:=
			U_{ijt}-U_{ijs},
			\quad
			F_{\Delta}(c)
			:=
			\Prob(\Delta_{ts}U_{ij}\le c)
			=
			\int F_U(c+u)\, dF_U(u).
			\]
			Then
			\[
			\sup_z\,
			\E\left[
			D_{ijt}(1-D_{ijs})
			\ind\{\Delta_{ts}W_{ij}(\theta_0)\le c\}
			\mid
			\mathcal Z_{ij}^{1:T}=z
			\right]
			\le
			F_{\Delta}(c)
			\]
			\[
			\le
			\inf_z
			\left[
			1-
			\E\left[
			(1-D_{ijt})D_{ijs}
			\ind\{\Delta_{ts}W_{ij}(\theta_0)\ge c\}
			\mid
			\mathcal Z_{ij}^{1:T}=z
			\right]
			\right].
			\]
			\item[(2)] For any balanced signed subgraph $(\mathcal C^+,\mathcal C^-)$ as in Proposition \ref{prop:signed-subgraph}, define
			\[
			\Delta_{\mathcal C}U
			:=
			\sum_{e\in\mathcal C^+} U_e
			-
			\sum_{e\in\mathcal C^-} U_e,
			\quad
			F_{\mathcal C}(c)
			:=
			\Prob(\Delta_{\mathcal C}U\le c).
			\]
			Then $F_{\mathcal C}$ is known from $F_U$; specifically, it is the convolution of $|\mathcal C^+|$ copies of $F_U$ and $|\mathcal C^-|$ copies of the reflected CDF $F_{-U}(x):=\Prob(-U_{ijt}\le x)$. Moreover,
			\[
			\sup_z
			\E\left[
			Y_{\mathcal C}^+
			\ind\{\Delta_{\mathcal C}W(\theta_0)\le c\}
			\mid
			\mathcal Z_{\mathcal C}^{1:T}=z
			\right]
			\le
			F_{\mathcal C}(c)
			\]
			\[
			\le
			\inf_z
			\left[
			1-
			\E\left[
			Y_{\mathcal C}^-
			\ind\{\Delta_{\mathcal C}W(\theta_0)\ge c\}
			\mid
			\mathcal Z_{\mathcal C}^{1:T}=z
			\right]
			\right].
			\]
		\end{enumerate}
	\end{proposition}

	\begin{remark}
		\label{rem:known-composite}
		Proposition \ref{prop:known-composite} is the most direct nonlogit sharpening available in the present setting. The gain comes from combining a known marginal CDF with serial independence, not from the marginal CDF alone. Once fixed effects are fully differenced out, the middle term in the lower/upper sandwich is pinned down by a known composite-error distribution. Logit remains distinct for a different reason: under additive node effects it yields an exact conditional-logit representation with algebraic cancellation of the fixed effects.
		
		There is also a useful max-score-type special case. Because $\Delta_{ts}U_{ij}=U_{ijt}-U_{ijs}$ is the difference of two i.i.d.\ continuous variables, its CDF is symmetric around zero and satisfies
		\[
		F_{\Delta}(0)=\frac{1}{2}.
		\]
		Hence Proposition \ref{prop:known-composite} implies
		\[
		\sup_z\,
		\E\left[
		D_{ijt}(1-D_{ijs})
		\ind\{\Delta_{ts}W_{ij}(\theta_0)\le 0\}
		\mid
		\mathcal Z_{ij}^{1:T}=z
		\right]
		\le
		\frac{1}{2},
		\]
		\[
		\sup_z\,
		\E\left[
		(1-D_{ijt})D_{ijs}
		\ind\{\Delta_{ts}W_{ij}(\theta_0)\ge 0\}
		\mid
		\mathcal Z_{ij}^{1:T}=z
		\right]
		\le
		\frac{1}{2}.
		\]
		This is the appropriate dynamic analog of the maximum-score-type special case in \citep{GaoWang26}, so it is best viewed as a sharpening of the dyad-transition bounds under serial independence, in the spirit of \citet{GaoWang26}, rather than as a separate identification result.
	\end{remark}
	
	\subsection{Additive Node Effects with Unknown CDF}
	
	Now suppose instead that the dyad effect is additive in nodes:
	\[
	A_{ij}=\nu_i+\nu_j.
	\]
	This assumption alone sharpens the semiparametric analysis because weighted differencing can now be organized around \emph{nodes} rather than \emph{dyads}. The admissible class of weighted configurations is therefore much larger than the dyad-balanced signed subgraphs used under unrestricted dyad effects.
	
	\begin{assumption}[Additive node effects and exchangeable node types]
		\label{ass:additive}
		Assume that (i) $A_{ij}=\nu_i+\nu_j$, (ii) the node histories $\{(\nu_i,Z_i^{1:T}): i=1,2,\ldots\}$ are i.i.d.\ across $i$, and (iii) the dyad-level shock process is independent of the full node history $\{(\nu_i,Z_i^{1:T}): i=1,\ldots,n\}$.
		
	\end{assumption}
	
	Let $\mathcal C$ be a finite nonempty collection of edge-time cells $e=(i,j,t)$, let $\omega_e\neq 0$ be an associated real weight, and let $\dot{e} = \{i,j\}$ be the set of nodes appearing in the dyad component of cell $e$. Define the positive and negative cells
	\[
	\mathcal C^+ := \{e\in\mathcal C:\omega_e>0\},
	\quad
	\mathcal C^- := \{e\in\mathcal C:\omega_e<0\},
	\]
	and, for each node $m$ appearing in $\mathcal C$, define its weighted incidence sum $\sigma_m := \sum_{e\in\mathcal C:\, m\in \dot{e}} \omega_e.$ Let $S_0 := \{m:\sigma_m=0\}$ be the set of nodes whose fixed effects are eliminated by the weighted configuration, and let $S_R := \{m:\sigma_m\neq 0\}$ be the set of retained nodes. Also let
	\[
	\mathcal Z_{S_0}^{1:T}
	:=
	\bigl({Z_m^{1:T}}'\bigr)'_{m\in S_0},\quad \mathcal Z_{S_R}^{1:T} := \bigl({Z_m^{1:T}}'\bigr)'_{m\in S_R}.
	\]
	Assume throughout that both $\mathcal C^+$ and $\mathcal C^-$ are nonempty. Define
	\[
	Y_{\mathcal C}^+
	:=
	\prod_{e\in\mathcal C^+} D_e
	\prod_{e\in\mathcal C^-} (1-D_e),
	\quad
	Y_{\mathcal C}^-
	:=
	\prod_{e\in\mathcal C^+} (1-D_e)
	\prod_{e\in\mathcal C^-} D_e,
	\]
	and write
	\[
	\Delta_{\mathcal C,\omega}W(\theta)
	:=
	\sum_{e\in\mathcal C}\omega_e W_e(\theta),
	\quad
	\widetilde U_{\mathcal C,\omega}
	:=
	\sum_{e\in\mathcal C}\omega_e U_e
	-
	\sum_{m\in S_R}\sigma_m \nu_m.
	\]
	
	\begin{proposition}[Weighted node-differencing under additive fixed effects]
		\label{prop:additive-weighted}
		Under \eqref{eq:main-model} and Assumptions \ref{ass:eh}-\ref{ass:additive}, for every $c\in\R$ and every realization $z_R$ of $\mathcal Z_{S_R}^{1:T}$, there exists a CDF $F_{\mathcal C,\omega}(\cdot\mid z_R)$ such that
		\[
		\sup_{z_0\in\Supp(\mathcal Z_{S_0}^{1:T}\mid \mathcal Z_{S_R}^{1:T}=z_R)}
		\E\left[
		Y_{\mathcal C}^+
		\ind\{\Delta_{\mathcal C,\omega}W(\theta_0)\le c\}
		\mid
		\mathcal Z_{S_R}^{1:T}=z_R,\;
		\mathcal Z_{S_0}^{1:T}=z_0
		\right]
		\le
		F_{\mathcal C,\omega}(c\mid z_R),
		\]
		and
		\[
		\inf_{z_0\in\Supp(\mathcal Z_{S_0}^{1:T}\mid \mathcal Z_{S_R}^{1:T}=z_R)}
		\left[
		1-
		\E\left[
		Y_{\mathcal C}^-
		\ind\{\Delta_{\mathcal C,\omega}W(\theta_0)\ge c\}
		\mid
		\mathcal Z_{S_R}^{1:T}=z_R,\;
		\mathcal Z_{S_0}^{1:T}=z_0
		\right]
		\right]
		\ge
		F_{\mathcal C,\omega}(c\mid z_R).
		\]
	\end{proposition}

	\begin{remark}
		\label{rem:additive-weighted}
		Proposition \ref{prop:additive-weighted} is the lagged-dynamic analog of the triad, weighted-star, and general cycle arguments in \citep*{GaoLiXu26}. It sharpens the unknown-CDF analysis even before specifying any parametric CDF. The key gain is combinatorial. First, complete elimination now requires only that weighted node incidences sum to zero, which is weaker than dyad balancing. Second, partial elimination is also admissible, because one conditions on the retained-node histories $\mathcal Z_{S_R}^{1:T}$ and profiles over the eliminated-node histories $\mathcal Z_{S_0}^{1:T}$, leaving the residual node effects inside the latent CDF $F_{\mathcal C,\omega}(\cdot\mid z_R)$. Additionally, dynamic versions of triads, weighted stars, tetrads, and longer cycles can all be used, and they can all contribute valid semiparametric restrictions.
		In particular, if one defines $\Theta_I^{\mathrm{add}}$ as the set of $\theta$ satisfying the envelope implications from Proposition \ref{prop:additive-weighted} for all admissible weighted configurations $(\mathcal C,\omega)$, all thresholds $c$, and all retained conditioning values $z_R$, then
		\[
		\theta_0 \in \Theta_I^{\mathrm{dyad}} \cap \Theta_I^{\mathrm{add}}.
		\]
		Thus additive node effects sharpen the semiparametric analysis even when the marginal CDF of $U_{ijt}$ is left unknown, so this sharpening is complementary to Proposition \ref{prop:known-composite}. It is worth noting precisely which components of Assumption~\ref{ass:additive} drive the result. The additive representation $A_{ij}=\nu_i+\nu_j$ enables the combinatorial gain of organizing weighted differencing around nodes rather than dyads. The i.i.d.\ node-history condition (Assumption~\ref{ass:additive}(ii)) and the stronger shock independence (Assumption~\ref{ass:additive}(iv)) are used in the proof to ensure that the retained node effects are conditionally independent of the eliminated-node histories given the retained histories, which is what allows profiling over $\mathcal Z_{S_0}^{1:T}$ while holding the latent CDF fixed. The proof does not use homogeneous marginals across dates, so this sharpening remains valid under arbitrary serial correlation within dyads. If, in addition, the assumptions of Proposition \ref{prop:known-composite} hold and the weighted configuration achieves complete node balance $\sigma_m=0$ for every node in $\mathcal C$, then $S_R$ is empty, $\widetilde U_{\mathcal C,\omega}=\sum_{e\in\mathcal C}\omega_e U_e$, and the now-unconditional CDF $F_{\mathcal C,\omega}$ is the known convolution of the scaled shock marginals.
	\end{remark}
	
	\subsection{Additive Node Fixed Effects with IID Logit Specification}
	
	The previous two subsections sharpened identification in two complementary directions: Section~\ref{Section: Sharper Identification under Additional Structures}.1 used a known marginal CDF with serial independence to make composite-error distributions explicit, while Section~\ref{Section: Sharper Identification under Additional Structures}.2 used additive node effects to enlarge the class of admissible differencing designs. This subsection combines the two strengthenings under a logit specification and shows that the combination yields an exact conditional logit representation that goes well beyond the per-period analogue of \citet{graham_2017}'s static tetrad logit. The key gain is that cross-node differencing (from Section~\ref{Section: Sharper Identification under Additional Structures}.2) and cross-period differencing (from Section~\ref{Section: Semiparametric Identification}) can be combined freely: any configuration of edge-time cells that achieves complete node balance produces an exact conditional logit, whether or not the cells share a common date. This yields a much larger class of identifying restrictions and a correspondingly weaker sufficient condition for point identification.
	
	We begin by motivating why logit is special. Suppose additive node effects are combined with a known conditional CDF $F$ for the current shock. At each date $t$, the model is
	\[
	D_{ijt}
	=
	\ind\left\{
	Z_{ijt}'\alpha_0 + X_{ij,t-1}'\lambda_0 + \nu_i + \nu_j - U_{ijt} \ge 0
	\right\}.
	\]
	Unlike the static strategic model studied in \citep*{GaoLiXu26}, there is no contemporaneous endogenous network statistic here, so the isolation machinery from that paper is not needed. If, conditional on the node effects and the lagged observables, the current shock on edge $(i,j)$ at date $t$ has CDF $F$, then
	\[
	p_{ij,t}
	:=
	\Prob\!\left(
	D_{ijt}=1
	\mid
	Z_{ijt},X_{ij,t-1},\nu
	\right)
	=
	F\!\left(Z_{ijt}'\alpha_0 + X_{ij,t-1}'\lambda_0 + \nu_i + \nu_j\right).
	\]
	For any configuration $\mathcal C=(\mathcal C^+,\mathcal C^-)$ of edge-time cells, the ratio $\Prob(Y_{\mathcal C}^+=1\mid \mathcal Z_{\mathcal C},\nu)/\Prob(Y_{\mathcal C}^-=1\mid \mathcal Z_{\mathcal C},\nu)$ involves terms of the form $F(\eta_e)/(1-F(\eta_e))$. The additive node effects cancel from the exponent $\sum_{e\in\mathcal C^+}\eta_e - \sum_{e\in\mathcal C^-}\eta_e$ if and only if $\sigma_m=0$ for every node $m$. But the multiplicative product of odds ratios reduces to an exponential of this sum if and only if $\log[F(\cdot)/(1-F(\cdot))]$ is affine---that is, up to location-scale normalization, exactly the logit case. For nonlogit $F$ (such as normal/probit), $\log[F(\cdot)/(1-F(\cdot))]$ is nonlinear and the node effects do not cancel algebraically from the product, so there is no exact conditional likelihood of the \citet{graham_2017} type.
	
	The semiparametric results above allow arbitrary serial correlation. The logit result below is sharper, but it does require a fully i.i.d.\ logistic shock structure.
	
	\begin{assumption}[IID logistic shocks with additive node effects]
		\label{ass:lt} Assume that (i) $A_{ij}=\nu_i+\nu_j$, and (ii) the shocks $\{U_{ijt}: i<j,\ t=1,\ldots,T\}$ are i.i.d.\ across dyads and dates with standard logistic CDF, and are jointly independent of the full latent-heterogeneity array and the full exogenous covariate array.
	\end{assumption}
	
	\noindent The standard logistic specification in Assumption~\ref{ass:lt}(ii) fixes both the location and scale of the error distribution, thereby resolving the scale indeterminacy present in the semiparametric analysis.
	
	Let $\mathcal C=(\mathcal C^+,\mathcal C^-)$ be a configuration of edge-time cells $e=(i,j,t)$, and recall the notation $\sigma_m=\#\{e\in\mathcal C^+: m\in\dot{e}\}-\#\{e\in\mathcal C^-: m\in\dot{e}\}$ for the signed incidence of node $m$. Say $\mathcal C$ is \emph{completely node-balanced} if $\sigma_m=0$ for every node $m$ appearing in $\mathcal C$. Define
	\[
	\Delta_{\mathcal C}W(\theta)
	:=
	\sum_{e\in\mathcal C^+}W_e(\theta)
	-
	\sum_{e\in\mathcal C^-}W_e(\theta),
	\]
	and let $\mathcal Z_{\mathcal C}$ denote the collection of observed exogenous histories for all edges appearing in $\mathcal C$.
	
	\begin{theorem}[Conditional logit under node-balanced comparisons]
		\label{thm:tetrad-logit}
		Under Assumption \ref{ass:lt}, let $\mathcal C=(\mathcal C^+,\mathcal C^-)$ be any completely node-balanced configuration. Then
		\[
		\log
		\frac{
			\Prob(Y_{\mathcal C}^+ = 1 \mid \mathcal Z_{\mathcal C})
		}{
			\Prob(Y_{\mathcal C}^- = 1 \mid \mathcal Z_{\mathcal C})
		}
		=
		\Delta_{\mathcal C}W(\theta_0).
		\]
		
		Equivalently,
		\[
		\Prob(Y_{\mathcal C}^+ =1 \mid Y_{\mathcal C}^+ + Y_{\mathcal C}^- = 1,\;\mathcal Z_{\mathcal C})
		=
		\frac{\exp(\Delta_{\mathcal C}W(\theta_0))}{1+\exp(\Delta_{\mathcal C}W(\theta_0))}.
		\]
		If the support of $\{\Delta_{\mathcal C}W(\theta_0):\mathcal C\ \text{completely node-balanced}\}$ spans $\R^{d_h+d_x}$, then $\theta_0=(\alpha_0',\lambda_0')'$ is point identified.
	\end{theorem}

	\begin{remark}[Tetrad logit as a special case]
		\label{rem:tetrad-special}
		The within-date tetrad is the simplest completely node-balanced configuration: for four distinct nodes $(i,j,h,k)$ and a single date $t$, set $\mathcal C^+=\{(i,j,t),(h,k,t)\}$ and $\mathcal C^-=\{(i,k,t),(j,h,t)\}$. Each node appears once in $\mathcal C^+$ and once in $\mathcal C^-$, so $\sigma_m=0$ for $m\in\{i,j,h,k\}$. In this case Theorem~\ref{thm:tetrad-logit} reduces to a per-period application of \citet{graham_2017}'s static tetrad logit with lagged regressors entering the index. That per-period result is not new: it follows directly from \citet{graham_2017} once the lagged network covariates are treated as predetermined. The contribution of Theorem~\ref{thm:tetrad-logit} is that it generalizes the tetrad logit by exploiting both the cross-node differencing of Section~\ref{Section: Sharper Identification under Additional Structures}.2 and the cross-period differencing of Section~\ref{Section: Semiparametric Identification}, thereby producing a substantially richer class of exact conditional logit restrictions.
	\end{remark}

	\begin{remark}[Examples of new configurations]
		\label{rem:new-configs}
		The following completely node-balanced configurations go beyond the within-date tetrad and are specific to the dynamic setting of this paper.
		
		\emph{Intertemporal tetrads.} Take four distinct nodes $(i,j,h,k)$ and (possibly distinct) dates $t_1,t_2,t_3,t_4$: set $\mathcal C^+=\{(i,j,t_1),(h,k,t_2)\}$ and $\mathcal C^-=\{(i,k,t_3),(j,h,t_4)\}$. Each node still appears once in $\mathcal C^+$ and once in $\mathcal C^-$. The covariate contrast $\Delta_{\mathcal C}W(\theta_0)$ now mixes cross-sectional and temporal variation, providing directions in $\R^{d_h+d_x}$ not available from any single-date tetrad.
		
		\emph{Triadic cycles.} Take three distinct nodes $\{i,j,k\}$ and six dates: set $\mathcal C^+=\{(i,j,t_1),(j,k,t_2),(i,k,t_3)\}$ and $\mathcal C^-=\{(i,j,t_4),(j,k,t_5),(i,k,t_6)\}$. Each node appears in exactly two edges of $\mathcal C^+$ and two edges of $\mathcal C^-$, so $\sigma_m=0$ for $m\in\{i,j,k\}$. This uses only three nodes rather than four, so triadic comparisons are available even in smaller networks.
		
		\emph{Longer cycles and weighted stars.} More generally, any node-balanced cycle of length $2k$ or any star configuration in which the hub's positive and negative incidences cancel produces an exact conditional logit. The class of such configurations grows combinatorially with the number of nodes and dates.
	\end{remark}
	
	\begin{remark}[Point identification: weakened support condition]
		\label{rem:support}
		The support condition in Theorem~\ref{thm:tetrad-logit} is stated over all completely node-balanced configurations, not just within-date tetrads. This is strictly weaker than requiring the within-date tetrad-differenced covariate vector to span $\R^{d_h+d_x}$, because intertemporal tetrads and triadic cycles contribute additional directions. The gain is substantive in at least two settings. First, when $n$ is small (so few tetrads exist at any given date), triadic cycles on three nodes expand the set of available comparisons. Second, when $X_{ij,t-1}$ contains count-valued network statistics such as common friends, its within-date tetrad difference takes integer values with limited variation, but intertemporal configurations pool across dates and can restore full rank.
		
		As a concrete illustration, consider $d_h=1$ and $X_{ij,t-1}=(D_{ij,t-1},R_{ij,t-1})'$ with $d_x=2$. Within-date tetrads generate covariate contrasts in $\R\times\mathbb{Z}^2$. The continuous covariate $Z_{it}$ provides variation along the first coordinate (though with possible point masses from the tetrad combination), and the integer-valued lagged-network differences provide the remaining two dimensions whenever the network is sufficiently heterogeneous. If the within-date tetrad support alone does not span $\R^3$, one can supplement it with intertemporal tetrads: at distinct dates $t_1\neq t_3$ or $t_2\neq t_4$, the lagged-network differences $\Delta X$ draw from different network configurations, and the resulting covariate contrasts can fill out missing directions.
	\end{remark}

	\begin{remark}[Moment inequality restrictions]
		\label{rem:logit-combined}
		Beyond the exact conditional logit of Theorem~\ref{thm:tetrad-logit}, the semiparametric results of Sections~\ref{Section: Semiparametric Identification}--\ref{Section: Sharper Identification under Additional Structures} also continue to apply under Assumption~\ref{ass:lt}. In particular, Assumption~\ref{ass:lt} implies both a known marginal CDF (standard logistic) and serial independence, so Proposition~\ref{prop:known-composite} gives explicit sandwich bounds for every dyad-balanced signed subgraph with the composite-error CDF computed as a known convolution of logistic distributions. Simultaneously, Proposition~\ref{prop:additive-weighted} applies, and for any completely node-balanced configuration the composite error has a known CDF. These moment inequality restrictions supplement the conditional logit of Theorem~\ref{thm:tetrad-logit} in two ways: they provide overidentifying restrictions useful for specification testing, and they supply additional identifying power through the ``bounding-by-$c$'' technique when the conditional logit support condition for point identification fails.
	\end{remark}

	\section{Conclusion}\label{Section: Conclusion}
	
	This paper studies a broad class of dynamic dyadic network formation models with time-varying observed covariates, lagged local network statistics, and unobserved heterogeneity. The framework nests observed-covariate homophily, transitivity, second-order or indirect-friend effects, and more general local subgraph statistics within a single dynamic index model. The main message is that, once these network covariates are lagged and observable, the model can be studied through a unified difference-out / integrate-out perspective rather than only through exact logit likelihood methods. Three principal strengthenings then sharpen that semiparametric analysis: a known marginal CDF combined with serial independence, additive node effects, and the special affine-log-odds structure of logit. Combining all three under i.i.d.\ logit with additive node effects yields an exact conditional logit representation for any completely node-balanced configuration of edge-time cells, generalizing the per-period analogue of \citet{graham_2017}'s tetrad logit by exploiting both cross-node and cross-period variation. Sharpness, inference, and further econometric development are left to future work.
	
	\bibliographystyle{ecta}
	\bibliography{references}
	
	\appendix
	
	\section{Proofs}
	
	\begin{proof}[Proof of Proposition \ref{prop:dyad}]
		
		Fix $h\in\Supp(Z_{ij}^{1:T})$ and $c\in\R$. For each date $t$, define
		\[
		F_h(c)
		:=
		\Prob(V_{ijt}\le c \mid Z_{ij}^{1:T}=h).
		\]
		This object does not depend on $t$. Indeed, using $V_{ijt}=U_{ijt}-A_{ij}$, the joint independence part of Assumption \ref{ass:eh}, and the law of iterated expectations,
		\[
		\begin{aligned}
			\Prob(V_{ijt}\le c \mid Z_{ij}^{1:T}=h)
			&=
			\E\!\left[
			\Prob(U_{ijt}\le c+A_{ij}\mid A_{ij},Z_{ij}^{1:T})
			\middle|\,
			Z_{ij}^{1:T}=h
			\right]
			\\
			&=
			\E\!\left[
			F_{U_t}(c+A_{ij})
			\middle|\,
			Z_{ij}^{1:T}=h
			\right],
		\end{aligned}
		\]
		where $F_{U_t}$ is the marginal CDF of $U_{ijt}$. By the homogeneous-marginal part of Assumption \ref{ass:eh}, $F_{U_t}$ is the same for every $t$, so the right-hand side is common across dates.
		
		Next, if $D_{ijt}=1$ and $W_{ijt}(\theta_0)\le c$, then by \eqref{eq:main-model},
		\[
		V_{ijt}\le W_{ijt}(\theta_0)\le c.
		\]
		Hence
		\[
		\ind\{D_{ijt}=1,\;W_{ijt}(\theta_0)\le c\}
		\le
		\ind\{V_{ijt}\le c\}.
		\]
		Taking expectations conditional on $Z_{ij}^{1:T}=h$ gives
		\[
		L_t(c\mid h;\theta_0)\le F_h(c).
		\]
		
		Similarly, if $D_{ijs}=0$ and $W_{ijs}(\theta_0)\ge c$, then
		\[
		V_{ijs}>W_{ijs}(\theta_0)\ge c,
		\]
		so
		\[
		\ind\{D_{ijs}=0,\;W_{ijs}(\theta_0)\ge c\}
		\le
		\ind\{V_{ijs}>c\}.
		\]
		Taking expectations conditional on $Z_{ij}^{1:T}=h$ yields
		\[
		\Prob\!\left(
		D_{ijs}=0,\;W_{ijs}(\theta_0)\ge c
		\mid
		Z_{ij}^{1:T}=h
		\right)
		\le
		1-F_h(c),
		\]
		and therefore
		\[
		F_h(c)\le U_s(c\mid h;\theta_0).
		\]
		Since this holds for every pair of dates $(t,s)$,
		\[
		L_t(c\mid h;\theta_0)\le F_h(c)\le U_s(c\mid h;\theta_0)
		\quad
		\text{for all } t,s.
		\]
		Taking the maximum over $t$ and the minimum over $s$ gives
		\[
		\overline L(c\mid h;\theta_0)\le \underline U(c\mid h;\theta_0).
		\]
		Because $c$ and $h$ were arbitrary, $\theta_0\in\Theta_I^{\mathrm{dyad}}$.
	\end{proof}
	
	\begin{proof}[Proof of Proposition \ref{prop:dyad-transition}]
		
		Fix $z\in\Supp(\mathcal Z_{ij}^{1:T})$ and $c\in\R$. If $D_{ijt}(1-D_{ijs})=1$, then by \eqref{eq:main-model},
		\[
		U_{ijt}\le W_{ijt}(\theta_0)+A_{ij},
		\quad
		U_{ijs}>W_{ijs}(\theta_0)+A_{ij}.
		\]
		Subtracting the second inequality from the first gives
		\[
		\Delta_{ts}U_{ij}
		<
		\Delta_{ts}W_{ij}(\theta_0).
		\]
		Therefore
		\[
		D_{ijt}(1-D_{ijs})
		\ind\{\Delta_{ts}W_{ij}(\theta_0)\le c\}
		\le
		\ind\{\Delta_{ts}U_{ij}<c\}.
		\]
		Taking expectations conditional on $\mathcal Z_{ij}^{1:T}=z$ and using the exogeneity part of Assumption \ref{ass:eh},
		\[
		\E\left[
		D_{ijt}(1-D_{ijs})
		\ind\{\Delta_{ts}W_{ij}(\theta_0)\le c\}
		\mid
		\mathcal Z_{ij}^{1:T}=z
		\right]
		\le
		\Prob(\Delta_{ts}U_{ij}<c).
		\]
		
		Likewise, if $(1-D_{ijt})D_{ijs}=1$, then
		\[
		U_{ijt}>W_{ijt}(\theta_0)+A_{ij},
		\quad
		U_{ijs}\le W_{ijs}(\theta_0)+A_{ij},
		\]
		so
		\[
		\Delta_{ts}U_{ij}
		>
		\Delta_{ts}W_{ij}(\theta_0).
		\]
		Hence
		\[
		(1-D_{ijt})D_{ijs}
		\ind\{\Delta_{ts}W_{ij}(\theta_0)\ge c\}
		\le
		\ind\{\Delta_{ts}U_{ij}>c\},
		\]
		and therefore
		\[
		\E\left[
		(1-D_{ijt})D_{ijs}
		\ind\{\Delta_{ts}W_{ij}(\theta_0)\ge c\}
		\mid
		\mathcal Z_{ij}^{1:T}=z
		\right]
		\le
		\Prob(\Delta_{ts}U_{ij}>c).
		\]
		
		Combining the two inequalities gives, for every $z$,
		\[
		\E\left[
		D_{ijt}(1-D_{ijs})
		\ind\{\Delta_{ts}W_{ij}(\theta_0)\le c\}
		\mid
		\mathcal Z_{ij}^{1:T}=z
		\right]
		\le
		\Prob(\Delta_{ts}U_{ij}<c)
		\le
		\Prob(\Delta_{ts}U_{ij}\le c),
		\]
		and
		\[
		\Prob(\Delta_{ts}U_{ij}\le c)
		=
		1-\Prob(\Delta_{ts}U_{ij}>c)
		\le
		1-
		\E\left[
		(1-D_{ijt})D_{ijs}
		\ind\{\Delta_{ts}W_{ij}(\theta_0)\ge c\}
		\mid
		\mathcal Z_{ij}^{1:T}=z
		\right].
		\]
		Taking the supremum over the first display and the infimum over the second yields the stated envelope inequality.
	\end{proof}
	
	\begin{proof}[Proof of Proposition \ref{prop:signed-subgraph}]
		
		Fix $c\in\R$. On the event $Y_{\mathcal C}^+ = 1$, each cell $(i,j,t)\in\mathcal C^+$ satisfies
		\[
		U_{ijt}\le W_{ijt}(\theta_0)+A_{ij},
		\]
		while each cell $(i,j,t)\in\mathcal C^-$ satisfies
		\[
		U_{ijt}>W_{ijt}(\theta_0)+A_{ij}.
		\]
		Subtracting the second collection from the first yields
		\[
		\Delta_{\mathcal C}U
		<
		\Delta_{\mathcal C}W(\theta_0)
		+
		\sum_{(i,j,t)\in\mathcal C^+}A_{ij}
		-
		\sum_{(i,j,t)\in\mathcal C^-}A_{ij}.
		\]
		By the balance condition, for each dyad $(i,j)$ the coefficient on $A_{ij}$ in the final two sums is
		\[
		\#\{t:(i,j,t)\in\mathcal C^+\}
		-
		\#\{t:(i,j,t)\in\mathcal C^-\}
		=
		0,
		\]
		so the dyad-effect term vanishes. Hence, on $Y_{\mathcal C}^+  = 1$,
		\[
		\Delta_{\mathcal C}U<\Delta_{\mathcal C}W(\theta_0).
		\]
		Therefore
		\[
		Y_{\mathcal C}^+ \ind\{\Delta_{\mathcal C}W(\theta_0)\le c\}
		\le
		\ind\{\Delta_{\mathcal C}U<c\}.
		\]
		
		Similarly, on the flipped event $Y_{\mathcal C}^- = 1$ one has
		\[
		\Delta_{\mathcal C}U>\Delta_{\mathcal C}W(\theta_0),
		\]
		and therefore
		\[
		Y_{\mathcal C}^- \ind\{\Delta_{\mathcal C}W(\theta_0)\ge c\}
		\le
		\ind\{\Delta_{\mathcal C}U>c\}.
		\]
		
		Now condition on $\mathcal Z_{\mathcal C}^{1:T}=z$. Because $\Delta_{\mathcal C}U$ is a measurable function of finitely many shocks and Assumption \ref{ass:eh} makes the shock process independent of the full exogenous covariate process, the conditional law of $\Delta_{\mathcal C}U$ does not depend on $z$. Hence, for every such $z$,
		\[
		\E\left[
		Y_{\mathcal C}^+ \ind\{\Delta_{\mathcal C}W(\theta_0)\le c\}
		\mid
		\mathcal Z_{\mathcal C}^{1:T}=z
		\right]
		\le
		\Prob(\Delta_{\mathcal C}U<c),
		\]
		and
		\[
		\E\left[
		Y_{\mathcal C}^- \ind\{\Delta_{\mathcal C}W(\theta_0)\ge c\}
		\mid
		\mathcal Z_{\mathcal C}^{1:T}=z
		\right]
		\le
		\Prob(\Delta_{\mathcal C}U>c).
		\]
		Consequently,
		\[
		\sup_z
		\E\left[
		Y_{\mathcal C}^+ \ind\{\Delta_{\mathcal C}W(\theta_0)\le c\}
		\mid
		\mathcal Z_{\mathcal C}^{1:T}=z
		\right]
		\le
		\Prob(\Delta_{\mathcal C}U<c)
		\le
		\Prob(\Delta_{\mathcal C}U\le c),
		\]
		while
		\[
		\Prob(\Delta_{\mathcal C}U\le c)
		=
		1-\Prob(\Delta_{\mathcal C}U>c)
		\le
		\inf_z
		\left[
		1-
		\E\left[
		Y_{\mathcal C}^-\ind\{\Delta_{\mathcal C}W(\theta_0)\ge c\}
		\mid
		\mathcal Z_{\mathcal C}^{1:T}=z
		\right]
		\right].
		\]
		This proves the claim.
	\end{proof}
	
	\begin{proof}[Proof of Proposition \ref{prop:partial-spectrum}]
		
		Fix $s\in\Supp(S_{\mathcal R})$ and $c\in\R$. Because both $\mathcal C_g^+$ and $\mathcal C_g^-$ are nonempty, for every $g\in\mathcal G$ the event $Y_g^+=1$ implies the strict inequality
		\[
		M_g<\Delta_gW(\theta_0),
		\]
		while $Y_g^-=1$ implies
		\[
		M_g>\Delta_gW(\theta_0).
		\]
		Consequently,
		\[
		Y_g^+\ind\{\Delta_gW(\theta_0)\le c\}
		\le
		\ind\{M_g<c\},
		\]
		and
		\[
		Y_g^-\ind\{\Delta_gW(\theta_0)\ge c\}
		\le
		\ind\{M_g>c\}.
		\]
		
		Let
		\[
		F(c\mid s)
		:=
		\Prob(M_g\le c\mid S_{\mathcal R}=s),
		\]
		where the right-hand side is well-defined and does not depend on the choice of $g\in\mathcal G$ by assumption iv. Because it is the conditional CDF of $M_g$ given $S_{\mathcal R}=s$, $F(\cdot\mid s)$ is a CDF.
		
		Now fix any $g\in\mathcal G$ and any $t\in\Supp(T_g\mid S_{\mathcal R}=s)$. Taking expectations conditional on $(S_{\mathcal R},T_g)=(s,t)$ gives
		\[
		\E\left[
		Y_g^+\ind\{\Delta_gW(\theta_0)\le c\}
		\mid
		S_{\mathcal R}=s,\;T_g=t
		\right]
		\le
		\Prob(M_g<c\mid S_{\mathcal R}=s,\;T_g=t).
		\]
		By assumption iv, the conditional law of $M_g$ given $S_{\mathcal R}=s$ does not depend on $T_g$, so
		\[
		\Prob(M_g<c\mid S_{\mathcal R}=s,\;T_g=t)
		=
		\Prob(M_g<c\mid S_{\mathcal R}=s)
		\le
		F(c\mid s).
		\]
		This proves the first inequality after taking the supremum over $g$ and $t$.
		
		Likewise,
		\[
		\E\left[
		Y_g^-\ind\{\Delta_gW(\theta_0)\ge c\}
		\mid
		S_{\mathcal R}=s,\;T_g=t
		\right]
		\le
		\Prob(M_g>c\mid S_{\mathcal R}=s,\;T_g=t)
		=
		1-F(c\mid s),
		\]
		again by assumption iv. Rearranging yields
		\[
		F(c\mid s)
		\le
		1-
		\E\left[
		Y_g^-\ind\{\Delta_gW(\theta_0)\ge c\}
		\mid
		S_{\mathcal R}=s,\;T_g=t
		\right].
		\]
		Taking the infimum over $g$ and $t$ proves the second inequality.
	\end{proof}
	
	\begin{proof}[Proof of Proposition \ref{prop:known-composite}]
		
		For part 1, Assumption \ref{ass:eh} and serial independence imply that $U_{ijt}$ and $U_{ijs}$ are independent and both have CDF $F_U$. Therefore, for every $c\in\R$,
		\[
		\begin{aligned}
			\Prob(\Delta_{ts}U_{ij}\le c)
			&=
			\Prob(U_{ijt}-U_{ijs}\le c)
			\\
			&=
			\int
			\Prob(U_{ijt}\le c+u\mid U_{ijs}=u)\,
			dF_U(u)
			\\
			&=
			\int F_U(c+u)\,dF_U(u),
		\end{aligned}
		\]
		which proves the expression for $F_{\Delta}$. Because $F_U$ is continuous, the difference $\Delta_{ts}U_{ij}$ has a continuous CDF, so
		\[
		\Prob(\Delta_{ts}U_{ij}<c)=F_{\Delta}(c),
		\quad
		\Prob(\Delta_{ts}U_{ij}>c)=1-F_{\Delta}(c).
		\]
		Applying Proposition \ref{prop:dyad-transition} yields the displayed dyad-transition bounds.
		
		For part 2, Assumption \ref{ass:eh} gives independence across dyads, while the added serial-independence assumption gives independence across dates within each dyad. Hence the shocks attached to distinct edge-time cells are mutually independent. It follows that $\Delta_{\mathcal C}U$ is the sum of $|\mathcal C^+|$ independent copies of $U_{ijt}$ and $|\mathcal C^-|$ independent copies of $-U_{ijt}$. Therefore its CDF is the convolution of $|\mathcal C^+|$ copies of $F_U$ and $|\mathcal C^-|$ copies of $F_{-U}$, and is known from $F_U$. Because each summand has a continuous CDF, $F_{\mathcal C}$ is continuous, so
		\[
		\Prob(\Delta_{\mathcal C}U<c)=F_{\mathcal C}(c),
		\quad
		\Prob(\Delta_{\mathcal C}U>c)=1-F_{\mathcal C}(c).
		\]
		Applying Proposition \ref{prop:signed-subgraph} then gives the stated signed-subgraph sandwich bounds.
	\end{proof}
	
	\begin{proof}[Proof of Proposition \ref{prop:additive-weighted}]
		
		Fix $c\in\R$ and a realization $z_R$ of $\mathcal Z_{S_R}^{1:T}$. For each edge-time cell $e=(i(e),j(e),t(e))$, write
		\[
		I_e(\theta_0)
		:=
		W_e(\theta_0)+\nu_{i(e)}+\nu_{j(e)}.
		\]
		On the event $Y_{\mathcal C}^+=1$, one has $D_e=1$ for every $e\in\mathcal C^+$ and $D_e=0$ for every $e\in\mathcal C^-$. Hence, by \eqref{eq:main-model},
		\[
		U_e\le I_e(\theta_0)
		\quad
		\text{for } e\in\mathcal C^+,
		\]
		and
		\[
		U_e>I_e(\theta_0)
		\quad
		\text{for } e\in\mathcal C^-.
		\]
		Multiplying the first collection by the positive weights $\omega_e>0$ and the second by the negative weights $\omega_e<0$ yields
		\[
		\omega_eU_e\le \omega_e I_e(\theta_0)
		\quad
		\text{for } e\in\mathcal C^+,
		\]
		and
		\[
		\omega_eU_e< \omega_e I_e(\theta_0)
		\quad
		\text{for } e\in\mathcal C^-.
		\]
		Summing over $e\in\mathcal C$ and using the strict inequality from $\mathcal C^-$ (which is nonempty) gives
		\[
		\sum_{e\in\mathcal C}\omega_e U_e
		<
		\sum_{e\in\mathcal C}\omega_e W_e(\theta_0)
		+
		\sum_m \sigma_m \nu_m
		=
		\Delta_{\mathcal C,\omega}W(\theta_0)+\sum_m \sigma_m \nu_m.
		\]
		Subtracting the non-eliminated node effects from both sides yields
		\[
		\widetilde U_{\mathcal C,\omega}
		<
		\Delta_{\mathcal C,\omega}W(\theta_0).
		\]
		Therefore
		\[
		Y_{\mathcal C}^+
		\ind\{\Delta_{\mathcal C,\omega}W(\theta_0)\le c\}
		\le
		\ind\{\widetilde U_{\mathcal C,\omega}< c\}.
		\]
		
		On the flipped event $Y_{\mathcal C}^-=1$, the inequalities reverse:
		\[
		U_e>I_e(\theta_0)
		\quad\text{for } e\in\mathcal C^+,
		\quad
		U_e\le I_e(\theta_0)
		\quad\text{for } e\in\mathcal C^-.
		\]
		After multiplication by the weights and summation, the resulting inequality is strict because $\mathcal C^+$ is nonempty:
		\[
		\widetilde U_{\mathcal C,\omega}
		>
		\Delta_{\mathcal C,\omega}W(\theta_0).
		\]
		Consequently,
		\[
		Y_{\mathcal C}^-
		\ind\{\Delta_{\mathcal C,\omega}W(\theta_0)\ge c\}
		\le
		\ind\{\widetilde U_{\mathcal C,\omega}>c\}.
		\]
		
		Now fix any
		\[
		z_0\in \Supp(\mathcal Z_{S_0}^{1:T}\mid \mathcal Z_{S_R}^{1:T}=z_R).
		\]
		Because $\widetilde U_{\mathcal C,\omega}$ is measurable with respect to the shock collection $\{U_e:e\in\mathcal C\}$ and the retained node effects $\{\nu_m:m\in S_R\}$, its conditional distribution given $(\mathcal Z_{S_R}^{1:T},\mathcal Z_{S_0}^{1:T})=(z_R,z_0)$ depends on the law of the retained node effects given $\mathcal Z_{S_R}^{1:T}=z_R$ and on the law of the shocks. By Assumption \ref{ass:additive}, the node histories are i.i.d.\ across nodes, so the retained node effects $\{\nu_m:m\in S_R\}$ are conditionally independent of the eliminated-node histories $\mathcal Z_{S_0}^{1:T}$ given $\mathcal Z_{S_R}^{1:T}$. By Assumption \ref{ass:eh}, the shock collection $\{U_e:e\in\mathcal C\}$ is independent of the full node-history array. Therefore
		\[
		\Prob\left(
		\widetilde U_{\mathcal C,\omega}\le c
		\mid
		\mathcal Z_{S_R}^{1:T}=z_R,\;
		\mathcal Z_{S_0}^{1:T}=z_0
		\right)
		=
		\Prob\left(
		\widetilde U_{\mathcal C,\omega}\le c
		\mid
		\mathcal Z_{S_R}^{1:T}=z_R
		\right).
		\]
		Denote this common conditional CDF by
		\[
		F_{\mathcal C,\omega}(c\mid z_R)
		:=
		\Prob\left(
		\widetilde U_{\mathcal C,\omega}\le c
		\mid
		\mathcal Z_{S_R}^{1:T}=z_R
		\right).
		\]
		
		Taking expectations conditional on
		\[
		(\mathcal Z_{S_R}^{1:T},\mathcal Z_{S_0}^{1:T})=(z_R,z_0)
		\]
		in the first indicator inequality yields
		\[
		\E\left[
		Y_{\mathcal C}^+
		\ind\{\Delta_{\mathcal C,\omega}W(\theta_0)\le c\}
		\mid
		\mathcal Z_{S_R}^{1:T}=z_R,\;
		\mathcal Z_{S_0}^{1:T}=z_0
		\right]
		\le
		F_{\mathcal C,\omega}(c\mid z_R).
		\]
		Since this bound holds for every admissible $z_0$, taking the supremum over $z_0$ proves the first displayed inequality in the proposition.
		
		Taking expectations conditional on
		\[
		(\mathcal Z_{S_R}^{1:T},\mathcal Z_{S_0}^{1:T})=(z_R,z_0)
		\]
		in the second indicator inequality yields
		\[
		\E\left[
		Y_{\mathcal C}^-
		\ind\{\Delta_{\mathcal C,\omega}W(\theta_0)\ge c\}
		\mid
		\mathcal Z_{S_R}^{1:T}=z_R,\;
		\mathcal Z_{S_0}^{1:T}=z_0
		\right]
		\le
		1-F_{\mathcal C,\omega}(c\mid z_R).
		\]
		Rearranging and then taking the infimum over admissible $z_0$ gives
		\[
		\inf_{z_0\in\Supp(\mathcal Z_{S_0}^{1:T}\mid \mathcal Z_{S_R}^{1:T}=z_R)}
		\left[
		1-
		\E\left[
		Y_{\mathcal C}^-
		\ind\{\Delta_{\mathcal C,\omega}W(\theta_0)\ge c\}
		\mid
		\mathcal Z_{S_R}^{1:T}=z_R,\;
		\mathcal Z_{S_0}^{1:T}=z_0
		\right]
		\right]
		\ge
		F_{\mathcal C,\omega}(c\mid z_R).
		\]
		This proves the proposition.
	\end{proof}
	
	\begin{proof}[Proof of Theorem \ref{thm:tetrad-logit}]
		
		Let $\mathcal C=(\mathcal C^+,\mathcal C^-)$ be a completely node-balanced configuration. Under Assumption \ref{ass:lt}, conditional on the node effects $\nu=(\nu_m)_m$ and the observed edge covariates $\mathcal Z_{\mathcal C}$, the link indicators $\{D_e:e\in\mathcal C^+\cup\mathcal C^-\}$ are independent Bernoulli random variables with success probabilities
		\[
		p_e
		=
		\Lambda(\eta_e),
		\quad
		\eta_e:=
		W_e(\theta_0) + \nu_{i(e)} + \nu_{j(e)},
		\]
		where $\Lambda(x)=\exp(x)/(1+\exp(x))$ is the logistic CDF. Hence
		\[
		\Prob(Y_{\mathcal C}^+ = 1 \mid \mathcal Z_{\mathcal C},\nu)
		=
		\prod_{e\in\mathcal C^+}\Lambda(\eta_e)
		\prod_{e\in\mathcal C^-}[1-\Lambda(\eta_e)],
		\]
		and
		\[
		\Prob(Y_{\mathcal C}^- =1 \mid \mathcal Z_{\mathcal C},\nu)
		=
		\prod_{e\in\mathcal C^+}[1-\Lambda(\eta_e)]
		\prod_{e\in\mathcal C^-}\Lambda(\eta_e).
		\]
		Taking the ratio and using $\Lambda(\eta)/[1-\Lambda(\eta)]=\exp(\eta)$ gives
		\[
		\frac{
			\Prob(Y_{\mathcal C}^+ =1 \mid \mathcal Z_{\mathcal C},\nu)
		}{
			\Prob(Y_{\mathcal C}^- = 1\mid \mathcal Z_{\mathcal C},\nu)
		}
		=
		\prod_{e\in\mathcal C^+}\exp(\eta_e)
		\prod_{e\in\mathcal C^-}\exp(-\eta_e)
		=
		\exp\!\left(
		\sum_{e\in\mathcal C^+}\eta_e - \sum_{e\in\mathcal C^-}\eta_e
		\right).
		\]
		Now substitute $\eta_e = W_e(\theta_0)+\nu_{i(e)}+\nu_{j(e)}$ and decompose the exponent:
		\[
		\sum_{e\in\mathcal C^+}\eta_e - \sum_{e\in\mathcal C^-}\eta_e
		=
		\Delta_{\mathcal C}W(\theta_0)
		+
		\sum_m \sigma_m\,\nu_m.
		\]
		By the complete node-balance assumption $\sigma_m=0$ for every node $m$, so
		\[
		\frac{
			\Prob(Y_{\mathcal C}^+ =1 \mid \mathcal Z_{\mathcal C},\nu)
		}{
			\Prob(Y_{\mathcal C}^- = 1\mid \mathcal Z_{\mathcal C},\nu)
		}
		=
		\exp\!\bigl(\Delta_{\mathcal C}W(\theta_0)\bigr).
		\]
		The right side depends only on observables and not on $\nu$, so taking conditional expectations with respect to $\nu$ given $\mathcal Z_{\mathcal C}$ preserves the same ratio:
		\[
		\Prob(Y_{\mathcal C}^+ = 1 \mid \mathcal Z_{\mathcal C})
		=
		\exp\!\bigl(\Delta_{\mathcal C}W(\theta_0)\bigr)\,
		\Prob(Y_{\mathcal C}^- = 1\mid \mathcal Z_{\mathcal C}).
		\]
		Taking logarithms gives the first display in the theorem. The conditional-logit representation follows by conditioning on $Y_{\mathcal C}^+ + Y_{\mathcal C}^- = 1$:
		\[
		\Prob(Y_{\mathcal C}^+ =1 \mid Y_{\mathcal C}^+ + Y_{\mathcal C}^- = 1,\;\mathcal Z_{\mathcal C})
		=
		\frac{\exp(\Delta_{\mathcal C}W(\theta_0))}{1+\exp(\Delta_{\mathcal C}W(\theta_0))}.
		\]
		
		For point identification, suppose $\theta$ satisfies $\Delta_{\mathcal C}W(\theta)=\Delta_{\mathcal C}W(\theta_0)$ for every completely node-balanced configuration $\mathcal C$. Then $\Delta_{\mathcal C}W(\theta-\theta_0)=0$ for every such $\mathcal C$. If the support of $\{\Delta_{\mathcal C}W(\theta_0)\}$ spans $\R^{d_h+d_x}$, this implies $\theta=\theta_0$.
	\end{proof}
	
	\section{Latent-Distance Dyad Effects}
	\label{sec:latent-distance}
	
	Now consider the more structured form
	\[
	A_{ij} = \nu_i + \nu_j - |\xi_i-\xi_j|,
	\]
	where $\xi_i$ is unobserved and time invariant. This case still fits Proposition \ref{prop:dyad} exactly. From the dyad-panel perspective,
	\[
	A_{ij}^\ast := \nu_i + \nu_j - |\xi_i-\xi_j|
	\]
	is just another time-invariant dyad effect, so the semiparametric identified-set arguments are unchanged.
	
	However, the node-balanced cancellation behind Theorem \ref{thm:tetrad-logit} now breaks. For instance, for a tetrad $(i,j,h,k)$,
	\[
	A_{ij}+A_{hk}-A_{ik}-A_{jh}
	=
	-|\xi_i-\xi_j|
	-|\xi_h-\xi_k|
	+|\xi_i-\xi_k|
	+|\xi_j-\xi_h|.
	\]
	The additive node effects still cancel, but the latent-distance terms do not generally vanish. The same failure applies to any completely node-balanced configuration.
	
	\begin{proposition}[What survives under latent-distance heterogeneity]
		\label{prop:case-c-basic}
		Under the model
		\[
		D_{ijt}
		=
		\ind\left\{
		Z_{ijt}'\alpha_0
		+ X_{ij,t-1}'\lambda_0
		+ \nu_i+\nu_j-|\xi_i-\xi_j|
		- U_{ijt}
		\ge 0
		\right\},
		\]
		Propositions \ref{prop:dyad}--\ref{prop:signed-subgraph} and Proposition \ref{prop:known-composite} remain valid. Proposition \ref{prop:additive-weighted} generally fails, and so does Theorem \ref{thm:tetrad-logit}.
	\end{proposition}
	
	\noindent No separate identification result is pursued here. Once $|\xi_i-\xi_j|$ is treated as part of a general time-invariant dyad effect, the semiparametric arguments of Propositions~\ref{prop:dyad}--\ref{prop:signed-subgraph} already cover it, and Proposition~\ref{prop:known-composite} applies whenever its serial-independence assumption holds. What is lost, relative to the additive-node benchmark, is the weighted node-differencing of Proposition~\ref{prop:additive-weighted} and the exact node-balanced cancellation behind Theorem~\ref{thm:tetrad-logit}, both of which require $A_{ij}=\nu_i+\nu_j$.
	
	\begin{proof}[Proof of Proposition \ref{prop:case-c-basic}]
		Define
		\[
		A_{ij}^\ast:=\nu_i+\nu_j-|\xi_i-\xi_j|.
		\]
		This object is time invariant at the dyad level. Therefore the model can be rewritten as
		\[
		D_{ijt}
		=
		\ind\left\{
		Z_{ijt}'\alpha_0+X_{ij,t-1}'\lambda_0+A_{ij}^\ast-U_{ijt}\ge 0
		\right\},
		\]
		which is of exactly the same form as \eqref{eq:main-model} with a generic time-invariant dyad effect. Proposition \ref{prop:dyad} uses only that time invariance, together with Assumption \ref{ass:eh}, so it applies without change. The same is true of Propositions \ref{prop:dyad-transition} and \ref{prop:signed-subgraph}: their proofs only use the fact that the same dyad effect appears with opposite signs and therefore cancels algebraically in the relevant signed comparison. Proposition~\ref{prop:known-composite} likewise applies, since it requires only Assumption~\ref{ass:eh}, serial independence, and a known marginal CDF, none of which depend on the form of $A_{ij}$.
		
		By contrast, Proposition~\ref{prop:additive-weighted} and Theorem~\ref{thm:tetrad-logit} both require the additive-node representation $A_{ij}=\nu_i+\nu_j$. The weighted node-differencing in Proposition~\ref{prop:additive-weighted} eliminates node effects $\nu_m$ by setting weighted incidence sums to zero, but the latent-distance component $|\xi_i-\xi_j|$ is a nonlinear function of node-level latent variables and generally does not cancel under the same weighted configuration. Similarly, any completely node-balanced configuration used in Theorem~\ref{thm:tetrad-logit} leaves a residual latent-distance term---for instance, the tetrad residual
		\[
		-(|\xi_i-\xi_j|+|\xi_h-\xi_k|)+|\xi_i-\xi_k|+|\xi_j-\xi_h|,
		\]
		which is generally nonzero and unobserved. Hence both Proposition~\ref{prop:additive-weighted} and Theorem~\ref{thm:tetrad-logit} generally fail.
	\end{proof}
	
\end{document}